\documentclass[conference,letterpaper]{IEEEtran}

\usepackage[margin=0.55in]{geometry}
\usepackage[utf8]{inputenc} 
\usepackage[T1]{fontenc}

\ifCLASSINFOpdf

\else

\fi
\usepackage[cmex10]{amsmath}
\usepackage{ifthen}
\usepackage{cite}
\usepackage{amsthm,amsfonts,amssymb}
\usepackage{graphicx}
\usepackage{subfig}
\usepackage{blkarray}
\usepackage[mathscr]{euscript}
\usepackage{enumitem}
\usepackage{algorithm}
\usepackage[noend]{algpseudocode}
\usepackage{blkarray}
\usepackage{stfloats}%
\newtheorem{theorem}{Theorem}[section]

\newtheorem{lemma}{Lemma}[section]
\theoremstyle{remark}
\newtheorem*{remark}{Remark}
\theoremstyle{definition}
\newtheorem{definition}{Definition}[section]
\usepackage{url}
\usepackage{arydshln}
\usepackage{mathtools}
\usepackage{dsfont}
\usepackage{physics}
\usepackage{tikz}
\usepackage{mathdots}
\usepackage{cancel}
\usepackage{color}
\usepackage{siunitx}
\usepackage{array}
\usepackage{multirow}
\usepackage{amssymb}
\usepackage{gensymb}
\usepackage{tabularx}
\usepackage{extarrows}
\usepackage{booktabs}
\usetikzlibrary{fadings}
\usetikzlibrary{patterns}
\usetikzlibrary{shadows.blur}
\usetikzlibrary{shapes}

\usepackage[cmintegrals]{newtxmath}

\begin{document}
%
\title{On the Performance of Reed-Muller Codes Over $(d,\infty)$-RLL Input-Constrained BMS Channels}
%
%
%
\author{%
	\IEEEauthorblockN{V.~Arvind Rameshwar}
	\IEEEauthorblockA{Indian Institute of Science, Bengaluru\\
		Email: \texttt{vrameshwar@iisc.ac.in}}
	\and
	\IEEEauthorblockN{Navin Kashyap}
	\IEEEauthorblockA{Indian Institute of Science, Bengaluru\\
		Email: \texttt{nkashyap@iisc.ac.in}
	}
}

\maketitle

\begin{abstract}
This paper considers the input-constrained binary memoryless symmetric (BMS) channel, without feedback. The channel input sequence respects the $(d,\infty)$-runlength limited (RLL) constraint, which mandates that any pair of successive $1$s be separated by at least $d$ $0$s. We consider the problem of designing explicit codes for such channels. In particular, we work with the Reed-Muller (RM) family of codes, which were shown by Reeves and Pfister (2021) to achieve the capacity of any unconstrained BMS channel, under bit-MAP decoding. We show that it is possible to pick $(d,\infty)$-RLL subcodes of a capacity-achieving (over the unconstrained BMS channel) sequence of RM codes such that the subcodes achieve, under bit-MAP decoding, rates of $C\cdot{2^{-\left \lceil \log_2(d+1)\right \rceil}}$, where $C$ is the capacity of the BMS channel. Finally, we also introduce  techniques for upper bounding the rate of any $(1,\infty)$-RLL subcode of a specific capacity-achieving sequence of RM codes. 
\end{abstract}


%
\IEEEpeerreviewmaketitle

\section{Introduction}
\label{sec:intro}
%
%
%
%
\IEEEPARstart{T}{he} physical limitations of hardware used in most data recording and communication systems cause some sequences to be more prone to error than others.  Constrained coding is a method of alleviating this problem, by encoding arbitrary user data sequences into sequences that respect a constraint (see, for example, \cite{Roth} or \cite{Immink}). In this work, we consider the problem of designing explicit constrained binary codes that achieve good rates of transmission over a noisy binary memoryless symmetric (BMS) channel. Examples of such channels include the binary erasure and binary symmetric channels (BEC and BSC, respectively) shown in Figures \ref{fig:bec} and \ref{fig:bsc}.

\begin{figure}[!h]
 \centering
\subfloat[]{
\resizebox{0.21\textwidth}{!}{

\tikzset{every picture/.style={line width=0.75pt}} 

\begin{tikzpicture}[x=0.75pt,y=0.75pt,yscale=-1,xscale=1]
	
	\draw    (454.24,136.67) -- (521.85,136.67) ;
	\draw [shift={(523.85,136.67)}, rotate = 180] [color={rgb, 255:red, 0; green, 0; blue, 0 }  ][line width=0.75]    (10.93,-3.29) .. controls (6.95,-1.4) and (3.31,-0.3) .. (0,0) .. controls (3.31,0.3) and (6.95,1.4) .. (10.93,3.29)   ;
	\draw    (454.24,136.67) -- (519.94,155.67) ;
	\draw [shift={(521.86,156.22)}, rotate = 196.13] [color={rgb, 255:red, 0; green, 0; blue, 0 }  ][line width=0.75]    (10.93,-3.29) .. controls (6.95,-1.4) and (3.31,-0.3) .. (0,0) .. controls (3.31,0.3) and (6.95,1.4) .. (10.93,3.29)   ;
	\draw    (454.24,173.79) -- (521.85,173.79) ;
	\draw [shift={(523.85,173.79)}, rotate = 180] [color={rgb, 255:red, 0; green, 0; blue, 0 }  ][line width=0.75]    (10.93,-3.29) .. controls (6.95,-1.4) and (3.31,-0.3) .. (0,0) .. controls (3.31,0.3) and (6.95,1.4) .. (10.93,3.29)   ;
	\draw    (454.24,173.79) -- (519.93,156.72) ;
	\draw [shift={(521.86,156.22)}, rotate = 165.43] [color={rgb, 255:red, 0; green, 0; blue, 0 }  ][line width=0.75]    (10.93,-3.29) .. controls (6.95,-1.4) and (3.31,-0.3) .. (0,0) .. controls (3.31,0.3) and (6.95,1.4) .. (10.93,3.29)   ;
	
	\draw (476.48,146) node [anchor=north west][inner sep=0.75pt]  [font=\footnotesize]  {$\epsilon $};
	\draw (476.48,158) node [anchor=north west][inner sep=0.75pt]  [font=\footnotesize]  {$\epsilon $};
	\draw (477.13,123.4) node [anchor=north west][inner sep=0.75pt]  [font=\footnotesize]  {$1-\epsilon $};
	\draw (477.13,175) node [anchor=north west][inner sep=0.75pt]  [font=\footnotesize]  {$1-\epsilon $};
	\draw (443.86,131.32) node [anchor=north west][inner sep=0.75pt]  [font=\footnotesize]  {${\displaystyle 0}$};
	\draw (443.86,166.96) node [anchor=north west][inner sep=0.75pt]  [font=\footnotesize]  {$1$};
	\draw (525.01,131.32) node [anchor=north west][inner sep=0.75pt]  [font=\footnotesize]  {$1$};
	\draw (525.01,149.88) node [anchor=north west][inner sep=0.75pt]  [font=\footnotesize]  {$0$};
	\draw (525.01,167.7) node [anchor=north west][inner sep=0.75pt]  [font=\footnotesize]  {$-1$};

\end{tikzpicture}
}
\label{fig:bec}
}
\qquad
\subfloat[]{
\resizebox{0.2\textwidth}{!}{

\tikzset{every picture/.style={line width=0.75pt}} 

\begin{tikzpicture}[x=0.75pt,y=0.75pt,yscale=-1,xscale=1]
	
	\draw    (449.24,119.67) -- (516.85,119.67) ;
	\draw [shift={(518.85,119.67)}, rotate = 180] [color={rgb, 255:red, 0; green, 0; blue, 0 }  ][line width=0.75]    (10.93,-3.29) .. controls (6.95,-1.4) and (3.31,-0.3) .. (0,0) .. controls (3.31,0.3) and (6.95,1.4) .. (10.93,3.29)   ;
	\draw    (449.24,119.67) -- (517.09,155.85) ;
	\draw [shift={(518.85,156.79)}, rotate = 208.07] [color={rgb, 255:red, 0; green, 0; blue, 0 }  ][line width=0.75]    (10.93,-3.29) .. controls (6.95,-1.4) and (3.31,-0.3) .. (0,0) .. controls (3.31,0.3) and (6.95,1.4) .. (10.93,3.29)   ;
	\draw    (449.24,156.79) -- (516.85,156.79) ;
	\draw [shift={(518.85,156.79)}, rotate = 180] [color={rgb, 255:red, 0; green, 0; blue, 0 }  ][line width=0.75]    (10.93,-3.29) .. controls (6.95,-1.4) and (3.31,-0.3) .. (0,0) .. controls (3.31,0.3) and (6.95,1.4) .. (10.93,3.29)   ;
	\draw    (449.24,156.79) -- (517.09,120.61) ;
	\draw [shift={(518.85,119.67)}, rotate = 151.93] [color={rgb, 255:red, 0; green, 0; blue, 0 }  ][line width=0.75]    (10.93,-3.29) .. controls (6.95,-1.4) and (3.31,-0.3) .. (0,0) .. controls (3.31,0.3) and (6.95,1.4) .. (10.93,3.29)   ;
	
	\draw (485.05,124) node [anchor=north west][inner sep=0.75pt]  [font=\footnotesize]  {$p$};
	\draw (483.48,144) node [anchor=north west][inner sep=0.75pt]  [font=\footnotesize]  {$p$};
	\draw (471.13,104.4) node [anchor=north west][inner sep=0.75pt]  [font=\footnotesize]  {$1-p$};
	\draw (470.13,159.4) node [anchor=north west][inner sep=0.75pt]  [font=\footnotesize]  {$1-p$};
	\draw (438.86,114.32) node [anchor=north west][inner sep=0.75pt]  [font=\footnotesize]  {${\displaystyle 0}$};
	\draw (438.86,149.96) node [anchor=north west][inner sep=0.75pt]  [font=\footnotesize]  {$1$};
	\draw (520.01,114.32) node [anchor=north west][inner sep=0.75pt]  [font=\footnotesize]  {$1$};
	\draw (520.01,150.7) node [anchor=north west][inner sep=0.75pt]  [font=\footnotesize]  {$-1$};

\end{tikzpicture}
}
\label{fig:bsc}}
\caption{(a) The binary erasure channel (BEC$(\epsilon)$) with erasure probability $\epsilon$ and output alphabet $\mathscr{Y} = \{-1,0,1\}$. (b) The binary symmetric channel (BSC$(p)$) with crossover probability $p$ and output alphabet $\mathscr{Y} = \{-1,1\}$.}
\end{figure}
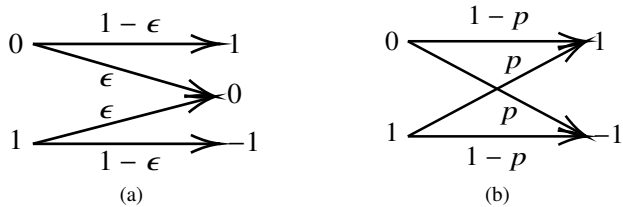
The input constraint of interest to us is the $(d,\infty)$-RLL constraint, which mandates that there are at least $d$ $0$s between every pair of successive $1$s in the input sequence. Figure 2 shows a state transition graph that represents the constraint. This constraint is a special case of the $(d,k)$-RLL constraint, which admits only binary sequences with at least $d$ and at most $k$ $0$s between successive $1$s. 
Reference \cite{Immink2} includes examples of $(d,k)$-RLL codes used in practice in magnetic storage and recording. 
\begin{figure}[!h]

\begin{center}
\resizebox{0.5\textwidth}{!}{
\tikzset{every picture/.style={line width=0.75pt}} 

\begin{tikzpicture}[x=0.75pt,y=0.75pt,yscale=-1,xscale=1]


\draw  [color={rgb, 255:red, 74; green, 144; blue, 226 }  ,draw opacity=1 ][line width=1.5]  (130,96) .. controls (130,82.19) and (141.19,71) .. (155,71) .. controls (168.81,71) and (180,82.19) .. (180,96) .. controls (180,109.81) and (168.81,121) .. (155,121) .. controls (141.19,121) and (130,109.81) .. (130,96) -- cycle ;
\draw  [color={rgb, 255:red, 74; green, 144; blue, 226 }  ,draw opacity=1 ][line width=1.5]  (230,96) .. controls (230,82.19) and (241.19,71) .. (255,71) .. controls (268.81,71) and (280,82.19) .. (280,96) .. controls (280,109.81) and (268.81,121) .. (255,121) .. controls (241.19,121) and (230,109.81) .. (230,96) -- cycle ;
\draw  [color={rgb, 255:red, 74; green, 144; blue, 226 }  ,draw opacity=1 ][line width=1.5]  (430,96) .. controls (430,82.19) and (441.19,71) .. (455,71) .. controls (468.81,71) and (480,82.19) .. (480,96) .. controls (480,109.81) and (468.81,121) .. (455,121) .. controls (441.19,121) and (430,109.81) .. (430,96) -- cycle ;
\draw  [color={rgb, 255:red, 74; green, 144; blue, 226 }  ,draw opacity=1 ][line width=1.5]  (530,96) .. controls (530,82.19) and (541.19,71) .. (555,71) .. controls (568.81,71) and (580,82.19) .. (580,96) .. controls (580,109.81) and (568.81,121) .. (555,121) .. controls (541.19,121) and (530,109.81) .. (530,96) -- cycle ;
\draw [line width=1.5]    (180,96) -- (226,96) ;
\draw [shift={(230,96)}, rotate = 180] [fill={rgb, 255:red, 0; green, 0; blue, 0 }  ][line width=0.08]  [draw opacity=0] (11.61,-5.58) -- (0,0) -- (11.61,5.58) -- cycle    ;
\draw [line width=1.5]    (280,96) -- (326,96) ;
\draw [shift={(330,96)}, rotate = 180] [fill={rgb, 255:red, 0; green, 0; blue, 0 }  ][line width=0.08]  [draw opacity=0] (11.61,-5.58) -- (0,0) -- (11.61,5.58) -- cycle    ;
\draw [line width=1.5]    (380,96) -- (426,96) ;
\draw [shift={(430,96)}, rotate = 180] [fill={rgb, 255:red, 0; green, 0; blue, 0 }  ][line width=0.08]  [draw opacity=0] (11.61,-5.58) -- (0,0) -- (11.61,5.58) -- cycle    ;
\draw [line width=1.5]    (480,96) -- (526,96) ;
\draw [shift={(530,96)}, rotate = 180] [fill={rgb, 255:red, 0; green, 0; blue, 0 }  ][line width=0.08]  [draw opacity=0] (11.61,-5.58) -- (0,0) -- (11.61,5.58) -- cycle    ;
\draw [line width=1.5]    (556,121) -- (556.5,185.73) ;
\draw [line width=1.5]    (156.03,125) -- (156.5,185.73) ;
\draw [shift={(156,121)}, rotate = 89.56] [fill={rgb, 255:red, 0; green, 0; blue, 0 }  ][line width=0.08]  [draw opacity=0] (11.61,-5.58) -- (0,0) -- (11.61,5.58) -- cycle    ;
\draw [line width=1.5]    (156.5,185.73) .. controls (155.5,212.73) and (167.5,213.73) .. (211.5,214.73) ;
\draw [line width=1.5]    (556.5,185.73) .. controls (557.5,213.73) and (551.5,214.73) .. (512.5,214.73) ;
\draw [line width=1.5]    (211.5,214.73) -- (328.5,214.73) ;
\draw [line width=1.5]    (386.5,214.73) -- (518.5,214.73) ;
\draw [shift={(382.5,214.73)}, rotate = 0] [fill={rgb, 255:red, 0; green, 0; blue, 0 }  ][line width=0.08]  [draw opacity=0] (11.61,-5.58) -- (0,0) -- (11.61,5.58) -- cycle    ;
\draw [line width=1.5]    (576.5,81.73) .. controls (622.56,60.17) and (626.36,138.5) .. (576.61,117.17) ;
\draw [shift={(573.5,115.73)}, rotate = 386.13] [fill={rgb, 255:red, 0; green, 0; blue, 0 }  ][line width=0.08]  [draw opacity=0] (11.61,-5.58) -- (0,0) -- (11.61,5.58) -- cycle    ;

\draw (344,96) node [anchor=north west][inner sep=0.75pt]  [font=\large]  {$\dotsc $};
\draw (344,210.7) node [anchor=north west][inner sep=0.75pt]  [font=\large]  {$\dotsc $};
\draw (150,88.4) node [anchor=north west][inner sep=0.75pt] [font=\Large]   {$0$};
\draw (250,88.4) node [anchor=north west][inner sep=0.75pt]   [font=\Large] {$1$};
\draw (438,88.4) node [anchor=north west][inner sep=0.75pt]   [font=\large] {$d-1$};
\draw (550,88.4) node [anchor=north west][inner sep=0.75pt]  [font=\Large]  {$d$};
\draw (198,73.4) node [anchor=north west][inner sep=0.75pt]  [font=\Large]  {$0$};
\draw (295,73.4) node [anchor=north west][inner sep=0.75pt]  [font=\Large]  {$0$};
\draw (395,73.4) node [anchor=north west][inner sep=0.75pt]  [font=\Large]  {$0$};
\draw (495,73.4) node [anchor=north west][inner sep=0.75pt]  [font=\Large]  {$0$};
\draw (615,89.4) node [anchor=north west][inner sep=0.75pt] [font=\Large]   {$0$};
\draw (543,143.4) node [anchor=north west][inner sep=0.75pt]  [font=\Large]  {$1$};
\draw (161,144.4) node [anchor=north west][inner sep=0.75pt] [font=\Large]   {$1$};

\end{tikzpicture}
}
\end{center}
\label{fig:d_inf}
\caption{The state transition graph for the $(d,\infty)$-RLL constraint.}
\end{figure}
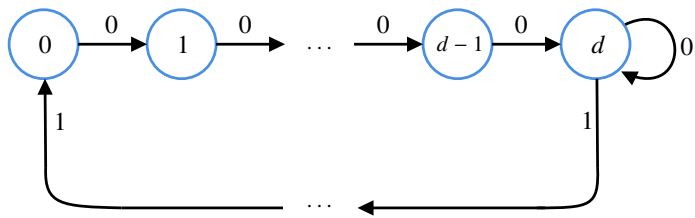
The system model under investigation in this paper considers input sequences that respect the $(d,\infty)$-RLL constraint, passed through the classical BMS channel, which is a special case of the discrete memoryless channel (DMC), introduced by Shannon in \cite{Sh48}. Figure \ref{fig:gen_const_DMC} shows a generic BMS channel with input constraints. Input-constrained DMCs in general fall under the broad class of discrete finite-state channels (DFSCs, or FSCs).

\begin{figure}[!t]
	\centering
	\resizebox{0.5\textwidth}{!}{

		\tikzset{every picture/.style={line width=0.75pt}} 
		
		\begin{tikzpicture}[x=0.75pt,y=0.75pt,yscale=-1,xscale=1]
			
			\draw   (321,106.32) -- (391,106.32) -- (391,174.32) -- (321,174.32) -- cycle ;
			\draw    (390.5,137.65) -- (438.5,137.65) ;
			\draw [shift={(440.5,137.65)}, rotate = 180] [color={rgb, 255:red, 0; green, 0; blue, 0 }  ][line width=0.75]    (10.93,-3.29) .. controls (6.95,-1.4) and (3.31,-0.3) .. (0,0) .. controls (3.31,0.3) and (6.95,1.4) .. (10.93,3.29)   ;
			\draw   (441,105) -- (511,105) -- (511,180) -- (441,180) -- cycle ;
			\draw   (561,117.65) -- (631,117.65) -- (631,157.65) -- (561,157.65) -- cycle ;
			\draw    (510.5,137.65) -- (558.5,137.65) ;
			\draw [shift={(560.5,137.65)}, rotate = 180] [color={rgb, 255:red, 0; green, 0; blue, 0 }  ][line width=0.75]    (10.93,-3.29) .. controls (6.95,-1.4) and (3.31,-0.3) .. (0,0) .. controls (3.31,0.3) and (6.95,1.4) .. (10.93,3.29)   ;
			\draw    (630.5,137.65) -- (678.5,137.65) ;
			\draw [shift={(680.5,137.65)}, rotate = 180] [color={rgb, 255:red, 0; green, 0; blue, 0 }  ][line width=0.75]    (10.93,-3.29) .. controls (6.95,-1.4) and (3.31,-0.3) .. (0,0) .. controls (3.31,0.3) and (6.95,1.4) .. (10.93,3.29)   ;
			\draw    (250.5,136.65) -- (318.5,136.65) ;
			\draw [shift={(320.5,136.65)}, rotate = 180] [color={rgb, 255:red, 0; green, 0; blue, 0 }  ][line width=0.75]    (10.93,-3.29) .. controls (6.95,-1.4) and (3.31,-0.3) .. (0,0) .. controls (3.31,0.3) and (6.95,1.4) .. (10.93,3.29)   ;

			\draw (461,151.01) node [anchor=north west][inner sep=0.75pt]  [font=\normalsize]  {$P_{Y|X}$};
			\draw (649,121.01) node [anchor=north west][inner sep=0.75pt]  [font=\normalsize]  {$\hat{m}$};
			\draw (529,120.01) node [anchor=north west][inner sep=0.75pt]  [font=\normalsize]  {$y^{n}$};
			\draw (409,121.01) node [anchor=north west][inner sep=0.75pt]  [font=\normalsize]  {$x^{n}$};
			\draw (571,132.61) node [anchor=north west][inner sep=0.75pt]  [font=\normalsize] [align=left] {Decoder};
			\draw (323,124) node [anchor=north west][inner sep=0.75pt]   [align=left] {{\normalsize Constrained}\\{\normalsize \ \ Encoder}};
			\draw (461,123) node [anchor=north west][inner sep=0.75pt]   [align=left] {BMS};
			\draw (255,116) node [anchor=north west][inner sep=0.75pt]  [font=\normalsize]  {${\displaystyle m\in \left[ 2^{nR}\right]}$};

		\end{tikzpicture}
}	
	\caption{System model of an input-constrained binary memoryless symmetric (BMS) channel without feedback.}
	\label{fig:gen_const_DMC}
\end{figure}
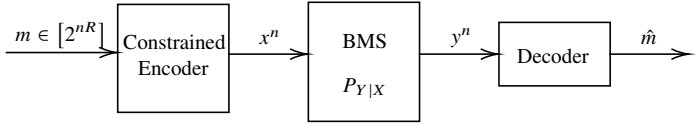

While explicit codes achieving the capacities or whose rates are very close to the capacity of unconstrained DMCs have been derived in works such as \cite{polar, kud1, luby, ru1, kud2}, the problem of designing coding schemes for input-constrained DMCs has not, to the best of our knowledge, been investigated in the literature. Moreover, unlike the case of the unconstrained DMC, whose capacity is characterized by Shannon's  single-letter, computable formula, $C_{\text{DMC}} = \sup_{P(x)} I(X;Y)$, the explicit computation of the capacity of an FSC, given by the maximum mutual information rate between inputs and outputs, is a much more difficult problem to tackle. In fact, the computation of the mutual information rate even for the simple case of Markov inputs reduces to the computation of the entropy rate of a Hidden Markov process---a well-known hard problem.

In this paper, we use the result of Reeves and Pfister \cite{Reeves} that Reed-Muller (RM) codes achieve the capacity of the unconstrained BMS channel under bit-MAP decoding, to design codes over our families of constrained BMS channels. We note that for the specific setting of the BEC, Kudekar et al.\ in \cite{kud1} were the first to show that Reed-Muller (RM) codes are capacity-achieving. Our approach to designing constrained codes over the BMS channel is simply to identify $(d,\infty)$-RLL subcodes of a sequence of capacity-achieving (over the unconstrained BMS channel) RM codes. Our results show that rates of $C\cdot{2^{-\left \lceil \log_2(d+1)\right \rceil}}$ are achievable over the $(d,\infty)$-RLL input-constrained BMS channels, where $C$ is the capacity of the unconstrained channel. Our results can be seen as accompanying the analysis in \cite{pvk}, on rates achievable by $(d,k)$-RLL subcodes of cosets of a linear block code. Specifically, from Corollary 1 of \cite{pvk}, we see that there exist cosets of  capacity-achieving (over the unconstrained BMS channel) codes, whose constrained subcodes have rate at least $C_0 + C -1$, where $C_0$ is the noiseless capacity of the input constraint. Our result provides an explicit sequence of codes of asymptotic rate that, for the high noise regimes of the BEC and BSC (large $\epsilon$, for the BEC, and $p$ close to 0.5, for the BSC), can be seen to be larger than the rate obtained using cosets of capacity-achieving codes in \cite{pvk}. The advantage of identifying such explicit codes is that low-complexity, off-the-shelf decoders can be employed to simplify the decoding process. 

Finally, we derive an upper bound on the rate of the largest $(1,\infty)$-RLL subcodes of the specific RM codes that we had used in our lower bounds. Our novel method of analysis uses properties of the weight distribution of RM codes---a topic that has received revived attention over the last decade (see, for example, the survey \cite{rm_survey} and the papers \cite{abbe1}, \cite{kaufman} and \cite{sberlo}). We hope that our techniques will prove useful in deriving upper bounds for other $(d,k)$- and $(d,\infty)$-RLL constraints, and will be extended, in future work, to any sequence of RM codes that is capacity-achieving over the unconstrained BMS channel.

The remainder of the paper is organized as follows: Section \ref{sec:notation} introduces the notation and refreshes some preliminary background. Section \ref{sec:main} states our main results. Section \ref{sec:rm} then describes the identification of subcodes of RM codes that achieve good rates over BMS channels and Section \ref{sec:rmub} discusses a technique that could be used to obtain an upper bound on the rate achievable over the BEC, using $(1,\infty)$-RLL subcodes of specific RM codes of constant rate. 
Finally, Section \ref{sec:conclusion} contains concluding remarks and a discussion on possible future work.

\section{Notation and Preliminaries}
\label{sec:notation}
\subsection{Notation}

Random variables will be denoted by capital letters, and their realizations by lower-case letters, e.g., $X$ and $x$, respectively. Calligraphic letters, e.g., $\mathscr{X}$, denote sets. The notation $[n]$ denotes the set, $\{1,2,\ldots,n\}$, of integers, and the notation $[a:b]$, for $a<b$, denotes the set of integers $\{a,a+1,\ldots,b\}$. Moreover, for a real number $x$, we use $\left \lceil x \right \rceil$ to denote the smallest integer larger than or equal to $x$. For vectors $\mathbf{w}$ and $\mathbf{v}$ of length $n$ and $m$, respectively, we denote their concatenation by the $(m+n)$-length vector, $\mathbf{w}\mathbf{v}$. The notation 
$x^N$ denotes the vector $(x_1,\ldots,x_N)$. Further, $P(x), P(y)$ and $P(y|x)$ are used to denote the probabilities $P_X(x), P_Y(y)$ and $P_{Y|X}(y|x)$, respectively, with the notation overloaded to refer to probability density functions in the case of continuous-valued random variables. The notation $X\sim \mathcal{N}(\mu, \sigma^2)$ refers to the fact that the random variable $X$ is drawn according to the Gaussian distribution, with mean $\mu$ and variance $\sigma^2>0$. 
Also, the notation $h_b(p):=-p\log_2 p - (1-p)\log_2(1-p)$ is the binary entropy function, for $p\in [0,1]$. All through, the empty summation is defined to be $0$, and the empty product is defined to be $1$. 
We write exp$_2(z)$ for $2^z$, where $z\in \mathbb{R}$. A logarithm to the base $2$ is denoted simply by $\log$, while the notation $\ln$ refers to the natural logarithm. Throughout, we use the convenient notation $\binom{m}{\le r}$ to denote the summation $\sum\limits_{i=0}^r \binom{m}{i}$.

\subsection{Problem Definition}

The communication setting of an input-constrained binary memoryless symmetric (BMS) channel without feedback is shown in Figure 4. A message $M$ is drawn uniformly from the set $\{1,2,\ldots,2^{nR}\}$, and is made available to the constrained encoder. The encoder produces a binary input sequence $x^n \in \{0,1\}^n = \mathscr{X}^n$, which is constrained to obey the $(d,\infty)$-RLL input constraint, a state transition graph for which is shown in Figure 2.

The channel output alphabet is the extended real line, i.e., $\mathscr{Y} = \overline{\mathbb{R}}$. The channel is memoryless in the sense that $P(y_i|x^{i},y^{i-1}) = P(y_i|x_i)$, for all $i$. Further, the channel is symmetric, in that $P(y|1) = P(-y|0)$, for all $y\in \mathscr{Y}$. Every such channel can be expressed as a multiplicative noise channel, in the following sense: if at any $i$ the input random symbol is $X_i\in \{0,1\}$, then the corresponding output symbol $Y_i\in \mathscr{Y}$ is given by
\[
Y_i = (-1)^{X_i}\cdot Z_i,
\]
where the noise random variables $Z^n$ are independent and identically distributed, and the noise process $(Z_i)_{i\geq 1}$ is independent of the input process $(X_i)_{i\geq 1}$. Common examples of such channels include the binary erasure channel (BEC$(\epsilon)$), with $P(Z_i = 1) = 1-\epsilon$ and $P(Z_i=0)=\epsilon$, the binary symmetric channel (BSC), with $P(Z_i = 1) = 1-p$ and $P(Z_i=-1)=p$, and the binary additive white Gaussian noise (BI-AWGN) channel, where $Z_i\sim \mathcal{N}(1,\sigma^2)$. Figures \ref{fig:bec} and \ref{fig:bsc} depict the BEC and BSC, pictorially. 
\begin{definition}
	\label{def:ach}
	An $(n,2^{nR},(d,\infty))$ code for an input-constrained channel {without feedback} is defined by the encoding function:
	\begin{equation}
		\label{eq:encoder}
		f: \{1,\ldots, 2^{nR}\}\rightarrow \mathscr{X}^n, \quad i\in [n],
	\end{equation}
	such that $(x_{i+1},\ldots,x_{\min\{i+d,n\}}) = (0,\ldots,0)$, if $x_i = 1$. 
	
	Given an output sequence $y^n$, the bit-MAP decoder $\Psi: \mathscr{Y}^n\rightarrow \mathscr{X}^n$ outputs $\hat{\mathbf{x}}:=(\hat{x}_1,\ldots,\hat{x}_n)$, where, for each $i\in [n]$, the estimate
	\begin{equation*}
		\hat{x}_i:=\text{argmax}_{x\in \{0,1\}}  P(X_i=x|y^n).
	\end{equation*}
	The error under bit-MAP decoding is defined as $$P_b^{(n)}:=1-\frac{1}{n} \sum_{i=1}^{n}\mathbb{E}[\max\{P(X_i=0|Y^n),P(X_i=1|Y^n)\}].$$
	A rate $R$ is said to be $(d,\infty)$-achievable under bit-MAP decoding, if there exists a sequence of $(n,2^{nR_n},(d,\infty))$ codes, $\{\mathcal{C}^{(n)}(R)\}_{n\geq 1}$, such that $\lim_{n\rightarrow \infty} P_b^{(n)} = 0$ and $\lim_{n\rightarrow \infty} R_n = R$. The sequence of codes $\{\mathcal{C}^{(n)}(R)\}$ is then said to achieve a rate of $R$ over the $(d,\infty)$-RLL input-constrained channel under bit-MAP decoding. The capacity, $C_{(d,\infty)}$, is defined to be the supremum over the respective $(d,\infty)$-achievable rates, and is a function of the parameters of the noise process. In this work, since we compute bounds on rates achievable using $(d,\infty)$-RLL subcodes of linear codes, we use the sub-optimal bit-MAP decoders for the larger linear codes, for decoding. Finally, a family of sequences of codes $\{\{\hat{\mathcal{C}}^{(n)}_{\mathbf{p}}\}_{n\geq 1}\}$ indexed by the noise parameters $\mathbf{p}$ is said to be \emph{capacity-achieving} (or $(d,\infty)$-capacity-achieving) if for all $\mathbf{p}$, $\{\hat{\mathcal{C}}^{(n)}_\mathbf{p}\}_{n\geq 1}$ achieves any rate $R\in (0,C_{(d,\infty)}(\mathbf{p}))$ over the $(d,\infty)$-RLL input-constrained channel. Similar definitions hold for the unconstrained (or $(0,\infty)$-RLL input-constrained) channel as well.
\end{definition}

%
%
%

\subsection{Reed-Muller Codes}
We recall the definition of the binary Reed-Muller (RM) family of codes. Codewords of binary RM codes consist of the evaluation vectors of multivariate polynomials over the binary field $\mathbb{F}_2$. Consider the polynomial ring $\mathbb{F}_2[x_1,x_2,\ldots,x_m]$ in $m$ variables. Note that in the specification of a polynomial $f\in \mathbb{F}_2[x_1,x_2,\ldots,x_m]$, only monomials of the form $\prod_{j\in S:S\subseteq [m]} x_j$ need to be considered, since $x^2 = x$ over the field $\mathbb{F}_2$, for an indeterminate $x$. For a polynomial $f\in \mathbb{F}_2[x_1,x_2,\ldots,x_m]$ and a binary vector $\mathbf{z} = (z_1,\ldots,z_m)\in \mathbb{F}_2^m$, let Eval$_\mathbf{z}(f):=f(z_1,\ldots,z_m)$. We let the evaluation points be ordered according to the standard lexicographic order on strings in $\mathbb{F}_2^m$, i.e., if $\mathbf{z} = (z_1,\ldots,z_m)$ and $\mathbf{z}^{\prime} = (z_1^{\prime},\ldots,z_m^{\prime})$ are two distinct evaluation points, then, $\mathbf{z}$ occurs before $\mathbf{z}^{\prime}$ in our ordering if and only if for some $i\geq 1$, it holds that $z_j = z_j^{\prime}$ for all $j<i$, and $z_i < z_i^{\prime}$. Now, let Eval$(f):=\left(\text{Eval}_\mathbf{z}(f):\mathbf{z}\in \mathbb{F}_2^m\right)$ be the evaluation vector of $f$, where the co-ordinates $\mathbf{z}$ are ordered according to the standard lexicographic order. 

\begin{definition}[see \cite{mws}, Chap. 13, or \cite{rm_survey}]
	The $r^{\text{th}}$ order binary Reed-Muller code RM$(m,r)$ is defined as the set of binary vectors:
	\[
	\text{RM}(m,r):=\{\text{Eval}(f): f\in \mathbb{F}_2[x_1,x_2,\ldots,x_m],\ \text{deg}(f)\leq r\},
	\]
	where $\text{deg}(f)$ is the degree of the largest monomial in $f$, and the degree of a monomial $\prod_{j\in S: S\subseteq [m]} x_j$ is simply $|S|$. 
\end{definition}

It is well-known that RM$(m,r)$ has dimension $\binom{m}{\le r} := \sum_{i=0}^{r}{m \choose i}$ and minimum Hamming distance $2^{m-r}$. The weight of a codeword $\mathbf{c} = \text{Eval}(f)$ is the number of $1$s in its evaluation vector, i.e,
\[
\text{wt}\left(\text{Eval}(f)\right):=|\{\mathbf{z}\in \mathbb{F}_2^m: f(\mathbf{z})=1\}|.
\]

The number of codewords in RM$(m,r)$ of weight $w$, for $w\in[2^{m-r}: 2^m]$, is given by the weight distribution function at $w$:
\[
A_{m,r}(w):=|\{\mathbf{c}\in \text{RM}(m,r): \text{wt}\left(\mathbf{c}\right) = w\}|.
\]
The subscripts $m$ and $r$ in $A_{m,r}$ will be suppressed when clear from context.
\section{Main Results}
\label{sec:main}
We now recall the main result of Reeves and Pfister in \cite{Reeves}. 
For a given $R\in (0,1)$, we consider the sequence of RM codes $\{\mathcal{C}_m(R)\}_{m\geq 1}$ under the lexicographic ordering of coordinates, where $\mathcal{C}_m(R) = \text{RM}(m,r_m)$, with 
\begin{equation}
	\label{eq:rmval}
	r_m := \max \left\{\left \lfloor \frac{m}{2}+\frac{\sqrt{m}}{2}Q^{-1}(1-R)\right \rfloor,0\right\},
\end{equation}
where $Q(\cdot)$ is the complementary cumulative distribution function (c.c.d.f.) of the standard normal distribution, i.e.,
\[
Q(t) = \frac{1}{\sqrt{2\pi}}\int_{t}^{\infty}e^{-\tau^2/2}d\tau, \ t\in \mathbb{R}.
\]

If $R_m$ is the rate of $\mathcal{C}_m(R)$, then, from Remark 24 in \cite{kud1}, it holds that $R_m\to R$ as $m\to \infty$. The following theorem then holds true:

\begin{theorem}[Theorem 1 of \cite{Reeves}]
	\label{thm:Reeves}
	Consider an unconstrained BMS channel with capacity $C\in (0,1)$. Then, any rate $R\in [0,C)$ is achieved by the sequence of codes $\{\mathcal{C}_m(R)\}$, under bit-MAP decoding.
\end{theorem}
As an example, Theorem \ref{thm:Reeves} implies that for the unconstrained BEC (resp. unconstrained BSC) with erasure probability $\epsilon \in (0,1)$ (resp. crossover probability $p\in (0,0.5)\cup (0.5,1)$), the sequence of codes $\{\mathcal{C}_m(1-\epsilon)\}_{m\geq 1}$ (resp. $\{\mathcal{C}_m(1-h_b(p))\}_{m\geq 1}$) achieves a rate of $1-\epsilon$ (resp. a rate of $1-h_b(p)$). Hence, the families of codes described above 
are $(0,\infty)$-capacity-achieving. 

Our idea is to construct a sequence of subcodes of $\{\mathcal{C}_m(R)\}$ that respect the $(d,\infty)$-RLL input-constraint, and analyze the rate of the chosen subcodes. We obtain the following result:

\begin{theorem}
	\label{thm:rm}
	For any $R \in (0,C)$, there exists a sequence of codes $\{\mathcal{C}_{m}^{(d,\infty)}(R)\}$, where $\mathcal{C}_{m}^{(d,\infty)}(R) \subset \mathcal{C}_m(R)$, which achieves a rate of $\frac{R}{2^{\left \lceil \log_2(d+1)\right \rceil}}$, over a $(d,\infty)$-RLL input-constrained BMS channel.
\end{theorem}

The proof of Theorem \ref{thm:rm} is provided in Section \ref{sec:rm}. 
Theorem \ref{thm:rm} states that for the $(d,\infty)$-RLL input-constrained BEC, a rate of $\frac{1}{d+1}(1-\epsilon)$ is achievable when $d=2^t-1$, for some $t\in \mathbb{N}$, and a rate of $\frac{1}{2(d+1)}(1-\epsilon)$ is achievable, otherwise. We note, however, that using random coding arguments, or using the techniques in \cite{ZW88} or \cite{arnk20}, it holds that a rate of $C_0^{(d)}(1-\epsilon)$ is achievable over the  $(d,\infty)$-RLL input-constrained BEC, where $C_0^{(d)}$ is the noiseless capacity of the input constraint (for example, $C_0^{(1)}\approx 0.6942$ and $C_0^{(2)}\approx 0.5515$). 
For the $(d,\infty)$-RLL input-constrained BSC, similarly, a rate of $\frac{1}{d+1}(1-h_b(p))$ is achievable when $d=2^t-1$, for some $t\in \mathbb{N}$, and a rate of $\frac{1}{2(d+1)}(1-h_b(p))$ is achievable, otherwise. Such a result is in the spirit of, but is weaker than, the conjecture by Wolf \cite{Wolf} that a rate of $C_0^{(d)}(1-h_b(p))$ is achievable over the $(d,\infty)$-RLL input-constrained BSC.


We now discuss a theorem that provides an upper bound on the largest rate achievable over a $(1,\infty)$-RLL input-constrained BMS channel, using subcodes of the sequence $\{\mathcal{C}_m(R) = \text{RM}(m,r_m)\}$, where $r_m$ is as in \eqref{eq:rmval}. 
Let $\mathcal{H}_{(1,\infty)}^{(m)}$ denote the largest subcode of ${\mathcal{C}}_m(R)$, all of whose codewords respect the $(1,\infty)$-RLL constraint. We then define
\[
\mathsf{R}^{(1,\infty)}_{{\mathcal{C}}}(R):=\limsup_{m\to \infty}\frac{\log_2|\mathcal{H}_{(1,\infty)}^{(m)}|}{2^m},
\]
to be the largest rate achieved by $(1,\infty)$-RLL subcodes of $\{{\mathcal{C}}_m(R)\}$, assuming that the ordering of the co-ordinates of the code is according to the standard lexicographic ordering. Then,
\begin{theorem}
	\label{thm:rmub}
	For the sequence of codes $\{{\mathcal{C}}_m(R) = \text{RM}(m,r_m)\}$, with $r_m$ as in \eqref{eq:rmval}, which achieve a rate $R$ over an unconstrained BMS channel, it holds that:
	\[
		\mathsf{R}^{(1,\infty)}_{{\mathcal{C}}}(R)\leq \frac{3R}{8}+\frac{1}{2}\ln \left(\frac{1}{1-R}\right),
		\]
		for $R\in (0,R^*)$, where $R^*\approx 0.37$ is the solution to: $\ln \left(\frac{1}{1-R}\right)  = \frac{5R}{4}$. For $R\geq R^*$, the trivial upper bound of $\mathsf{R}^{(1,\infty)}_{{\mathcal{C}}}(R)\leq R$ holds.
\end{theorem}
The proof of the theorem is taken up in Section \ref{sec:rmub}. Figure \ref{fig:ubplot} shows a comparison between the upper bound in Theorem \ref{thm:rmub} and the lower bound from Theorem \ref{thm:rm}.

\begin{figure*}[!t]
	\centering
	\includegraphics[width = 0.8\textwidth]{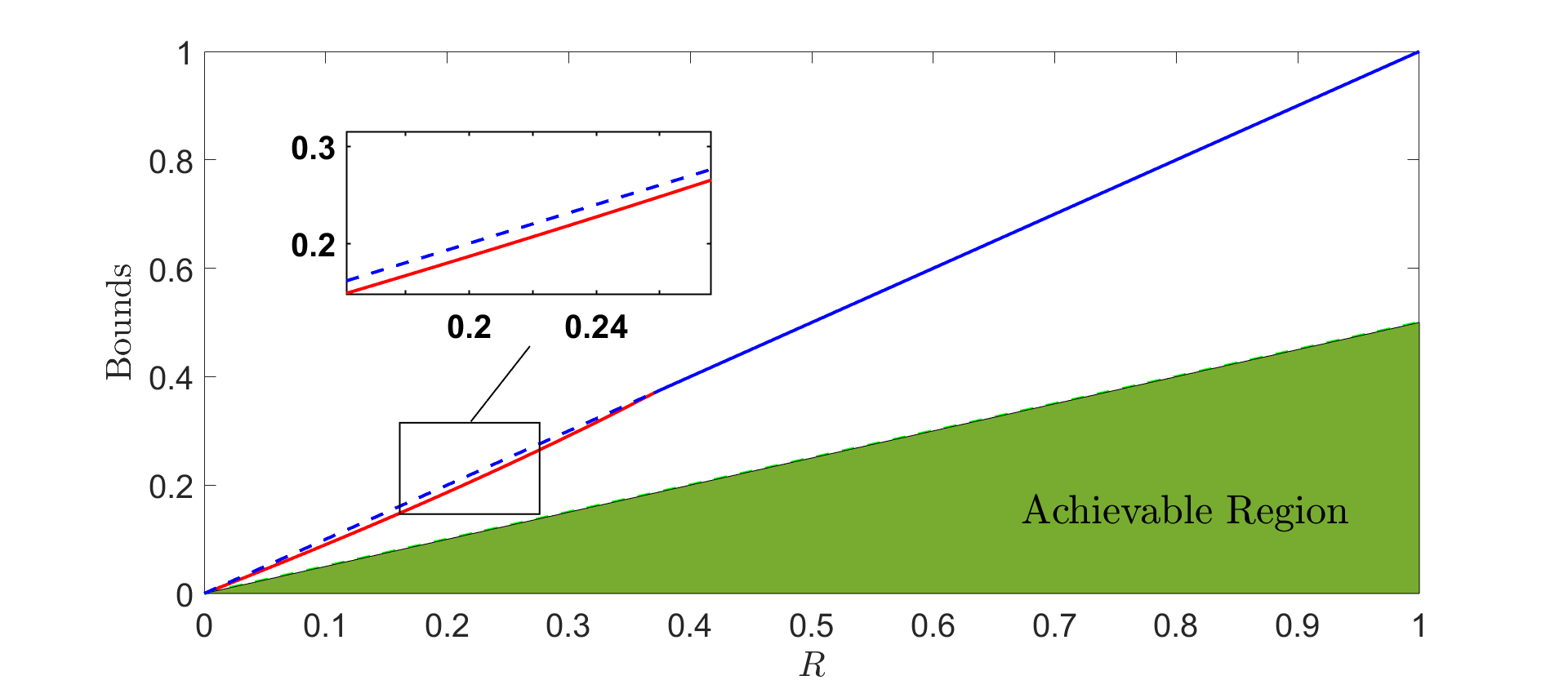}
	\caption{A comparison between the upper bound and achievable rate, using subcodes of RM codes, computed for the $(1,\infty)$-RLL input-constrained BEC. The upper bound is shown by the thick red and blue curves, while the achievable region from Theorem \ref{thm:rm} is shown in green. A comparison between the upper bound in Theorem \ref{thm:rmub} and the trivial upper bound of $R$ is shown in the inset, for $R\in (0,R^*)$.}
	\label{fig:ubplot}
\end{figure*}
We end this section with a couple of remarks. Firstly, note that the all-ones codeword $\mathbf{1}$ belongs to the RM code. Since any codeword $\mathbf{c}$ that respects the $(0,1)$-RLL constraint can be written as $\mathbf{c} = \mathbf{1}+\mathbf{\hat{c}}$, where $\mathbf{\hat{c}}$ respects the $(1,\infty)$-RLL constraint, the lower bound of Theorem \ref{thm:rm} and the upper bound of Theorem \ref{thm:rmub} hold for the rate of $(0,1)$-RLL subcodes as well. Moreover, since for any $k>1$, a $(0,1)$-RLL subcode of an RM code is a subset of a $(0,k)$-RLL subcode, the lower bound of $R/2$ from Theorem \ref{thm:rm} holds for $(0,k)$-RLL subcodes of RM codes as well.

Secondly, in this work, we are concerned with $(d,\infty)$- or $(0,k)$-RLL subcodes of a linear (Reed-Muller) code. We point out that the dual problem of identifying linear codes that are subsets of the set of $(d,\infty)$- or $(0,k)$-RLL sequences of a fixed length, has also been studied \cite{lechner}. The results there can be suitably extended to show that the rate of the largest linear code within the set of $(d,\infty)$- or $(0,k)$-RLL sequences of length $n$, equal, respectively, $\frac{1}{d+1}$ and $\frac{k}{k+1}$, as $n\to \infty$. However, such a result offers no insight into rates achievable over BMS channels.
\section{Achievable Rates}
\label{sec:rm}

As mentioned in Section \ref{sec:main}, we consider the Reed-Muller (RM) family of codes, $\{\mathcal{C}_m(R)\}$, which are such that any rate $R\in [0,C)$ is $(0,\infty)$-achievable (see Definition \ref{def:ach}) over the BMS channel (see Theorem \ref{thm:Reeves}), where $C$ denotes the capacity of the channel. We then select subcodes that respect the $(d,\infty)$-RLL constraint, of such a capacity-achieving sequence of RM codes, and compute their rate. For notational convenience, we denote by $S_{(d,\infty)}^{(n)}$, the set of $n$-length binary vectors that respect the $(d,\infty)$-RLL constraint, i.e.,
\begin{align*}
S_{(d,\infty)}^{(n)}:=\bigl\{\mathbf{c} = (c_0,\ldots,c_{n-1}): \mathbf{c}&\text{ respects the $(d,\infty)$-RLL}\\ &\text{ input constraint}\bigr\}.
\end{align*}
We suppress the superscript `$(n)$' if the length of the vector is clear from the context.

We begin with a simple observation, presented below as a lemma. We recall that the support of a vector $\mathbf{c}\in \mathbb{F}_2^{n}$,
\[
\text{supp}(\mathbf{c}) = \{i: c_i = 1\}.
\]

\begin{lemma}
	\label{lem:simplerm}
	Given $d\geq 1$, if $\hat{\mathbf{c}}$ is such that $\hat{\mathbf{c}}\in S_{(d,\infty)}$, and supp$(\mathbf{c}) \subseteq \text{supp}(\hat{\mathbf{c}})$, then it holds that $\mathbf{c}\in S_{(d,\infty)}$.
\end{lemma}


Another fact that we repeatedly use is recorded in Lemma~\ref{lem:rate} below. For this, recall our definition of $r_m$ from \eqref{eq:rmval}.

\begin{lemma}
\label{lem:rate}
For any sequence of positive integers $(t_m)$ such that $t_m = o(\!\!\sqrt{m})$, we have
$$
\lim_{m \to \infty} \frac{1}{2^{m-t_m}} \binom{m-t_m}{\le r_m} \ = \ R.
$$
In particular, for any fixed integer $t > 0$, $\lim\limits_{m \to \infty} \frac{1}{2^m} \binom{m-t}{\le r_m} = 2^{-t}R$.
\end{lemma}
\begin{proof}
Let $S_m$ denote a Bin$(m,\frac12)$ random variable, and note that $\frac{1}{2^{m-t_m}} \binom{m-t_m}{\le r_m}$ equals $\Pr[S_{m-t_m} \le r_m]$. Further note that by our choice of $r_m$, for any integer $t > 0$, we have for all $m$ large enough,
\begin{align}
	|r_m - r_{m-t}| &\leq \biggl| \frac{m}{2} + \frac{\sqrt{m}}{2}Q^{-1}(1-R) \, - \notag\\
	   &\ \ \ \ \ \ \ \ \ \ \ \left(\frac{m-t}{2} + \frac{\sqrt{m-t}}{2}Q^{-1}(1-R)\right) + 1 \biggr| \notag\\
	&\leq \frac{t}{2}+\frac{\sqrt{t}}{2}\lvert Q^{-1}(1-R)\rvert + 1. \label{eq:rm_diff}
\end{align}
Hence, we have $r_{m-t_m} - \nu_m \le r_m \le r_{m-t_m} + \nu_m$, with $\nu_m := \frac{t_m}{2}+\frac{\sqrt{t_m}}{2}\lvert Q^{-1}(1-R)\rvert + 1$.
Consequently, $\Pr[S_{m-t_m} \le r_{m-t_m} - \nu_m] \le \Pr[S_{m-t_m} \le r_m] \le \Pr[S_{m-t_m} \le r_{m-t_m} + \nu_m]$. Setting $\overline{S}_{m-t_m} := \frac{S_{m-t_m} - \frac12(m-t_m)}{\frac12\sqrt{m-t_m}}$, we have 
\begin{align}
& \Pr[\overline{S}_{m-t_m} \le Q^{-1}(1-R) - \frac{\nu_m}{\frac12\sqrt{m-t_m}}] \notag \\
& \ \ \ \ \ \le \ \Pr[S_{m-t_m} \le r_m] \notag \\
& \ \ \ \ \ \ \ \ \ \ \ \le \ \Pr[\overline{S}_{m-t_m} \le Q^{-1}(1-R) + \frac{\nu_m}{\frac12\sqrt{m-t_m}}].
\label{sandwich}
\end{align}

Now, by the central limit theorem (or, in this special case, by the de Moivre-Laplace theorem), $\overline{S}_{m-t_m}$ converges in distribution to a standard normal random variable, $Z$. Therefore, via \eqref{sandwich} and the fact that $t_m$ and $\nu_m$ are both $o(\!\!\sqrt{m})$, we obtain that 
$$
\lim_{m \to \infty} \Pr[S_{m-t_m} \le r_m]  \ = \ \Pr[Z \le Q^{-1}(1-R)] \ = \ R,
$$
which proves the lemma.
\end{proof}

We are now in a position to prove Theorem~\ref{thm:rm}.

\begin{proof}[Proof of Theorem~\ref{thm:rm}]
	For a fixed $d\geq 1$, let $z := \left \lceil \log_2(d+1)\right \rceil$. Consider the subcode $\mathcal{C}_m^{(d,\infty)}(R)$, of the code $\mathcal{C}_m(R)$, defined as:
	\begin{align*}
	\mathcal{C}_m^{(d,\infty)}(R):=\Bigg\{\text{Eval}(f): f = &\bigg(\prod_{i=m-z+1}^{m} x_i \bigg)\cdot g(x_{1},\ldots, x_{m-z}),\\ &\text{ where } \text{deg}(g)\leq r_m-z\Bigg\}.
	\end{align*}
It is easy to verify that the polynomial $h(x_{m-z+1},\ldots,x_m):=\prod_{i=m-z+1}^{m} x_i$ is such that its corresponding evaluation vector, Eval$(h)$, obeys Eval$(h)\in S_{(d,\infty)}^{(2^m)}$. This is because Eval$_\mathbf{y}(h) = 1$ if and only if $(y_{m-z+1},\ldots,y_m) = (1,\ldots,1)$, and in the lexicographic ordering, such evaluation points $\mathbf{y}$ are spaced out by $2^z - 1$ coordinates, where $2^z-1\geq d$. Thus, for any polynomial $f$ such that $\text{Eval}(f)\in \mathcal{C}_m^{(d,\infty)}(R)$, it is true that supp$(\text{Eval}(f))\subseteq \text{supp}(\text{Eval}(h))$, so that Eval$(f) \in S_{(d,\infty)}^{(2^m)}$ via Lemma \ref{lem:simplerm}. Hence, for any $R\in (0,1)$, it holds that $\mathcal{C}_m^{(d,\infty)}(R)$ is a collection of RM codewords in $\mathcal{C}_m(R)$, all of which respect the input constraint.

Consequently, the rate of the subcode $\mathcal{C}_m^{(d,\infty)}(R)$ is
\begin{align*}
	R_m^{(d,\infty)} \ & =\ \frac{\log_2(|\mathcal{C}_m^{(d,\infty)}|)}{2^m} \\ 
	&= \ \frac{{m-z \choose \leq r_m-z}}{2^m} 
		\ = \ \frac{\binom{m-z}{\le r_m-z}}{\binom{m-z}{\le r_m}} \, \frac{{m-z \choose \leq r_m}}{2^m}
		\ \xrightarrow{m \to \infty} \ 2^{-z}R.
\end{align*}
 To obtain the limit as $m \to \infty$, we have used Lemma~\ref{lem:rate} and the fact that the ratio $\frac{\binom{m-z}{\le r_m-z}}{\binom{m-z}{\le r_m}}$ converges to $1$ as $m \to \infty$.
\end{proof}
\begin{remark}
	The achievable rate calculated in this paper supplements the results of \cite{pvk}, which show the existence of $(d,\infty)$- and $(0,k)$-RLL) subcodes of cosets of any $(0,\infty)$-capacity-achieving linear code, over a BMS channel with capacity $C$, of rate at least $C_0 + C-1$, where $C_0$ is the noiseless capacity of the input constraint. In particular, the paper shows the existence of codes of rate at least $C_0-\epsilon$, for the BEC$(\epsilon)$, and $C_0 - h_b(p)$, for the BSC$(p)$. However, the paper does not provide an explicit identification of such codes. We note in addition that our achievable rate exceeds the lower bounds in \cite{pvk}, for large $\epsilon$, over the BEC, and for $p$ close to $0.5$, over the BSC.
\end{remark}
In the context of the design of coding schemes over (input-constrained) BMS channels, it would be remiss to not comment on the rates achieved by polar codes, given that polar codes are capacity-achieving over a broad class of channels (see, for example, \cite{polar}, \cite{tal}, \cite{tal2}, and reference therein). Following the work of Li and Tan in \cite{litan}, and as pointed out earlier, it holds that the capacity without feedback of the class of input-constrained DMCs can be approached arbitrarily closely using stationary, ergodic Markov input distributions of finite order. Moreover, from \cite{tal}, it can be verified that polar codes do achieve the information rates over the BEC, of any stationary, ergodic finite-state Markov input process, that satisfy a certain ``forgetfulness'' property, a sufficient condition for which is given by Condition K in the paper. In particular, this shows that polar codes do achieve the capacity of the $(d,\infty)$- and $(0,k)$-RLL input-constrained BEC. However, this observation is not very helpful for the following reasons:
\begin{itemize}
	\item We do not possess knowledge of an optimal sequence of Markovian input distributions.
	\item Polar codes are (structured) random codes, and the exact choice of the bit-channels to send information bits over, requires analysis. In our construction, however, we have explicitly identified the codewords being transmitted.
	\item Finally, this observation does not lend insight into explicit lower bounds on non-feedback capacity.
\end{itemize}
\section{Upper Bounds}
\label{sec:rmub}
In this section, we provide upper bounds on the rates over a $(1,\infty)$-RLL input-constrained BMS channel, achievable using subcodes of $\{\mathcal{C}_m(R)\}$ (see \eqref{eq:rmval}), which is a $(0,\infty)$-capacity-achieving family of RM codes. We fix the co-ordinate ordering to be the standard lexicographic ordering. 

Following the expositions in \cite[Chap. 13]{mws} and \cite{rm_survey}, it holds that any Boolean polynomial $f \in \mathbb{F}_2[x_1,\ldots,x_m]$, such that Eval$(f) \in \text{RM}(m,r)$, can be expressed as:
\begin{equation}
	\label{eq:rmdecomp}
f(x_1,\ldots,x_m) = g(x_1,\ldots,x_{m-1}) + x_m\cdot h(x_1,\ldots,x_{m-1}), 
\end{equation}
where $g,h$ are such that Eval$(g)\in \text{RM}(m-1,r)$ and Eval$(h)\in \text{RM}(m-1,r-1)$. The following lemma then holds true:
\begin{lemma}
	\label{lem:necesscond}
	If a codeword Eval$(f)\in \text{RM}(m,r)$ is such that Eval$(f)\in S_{(1,\infty)}$, then, supp$(\text{Eval}(g)) \subseteq \text{supp}(\text{Eval}(h))$, where $g,h$ are as in \eqref{eq:rmdecomp}.
\end{lemma}
\begin{proof}
	Suppose that there exists some evaluation point $\mathbf{z} = (z_1,\ldots,z_{m-1}) \in \mathbb{F}_2^{m-1}$ such that $g(\mathbf{z})=1$ and $h(\mathbf{z})=0$. Then it follows that at evaluation points $\mathbf{z}_1,\mathbf{z}_2 \in \mathbb{F}_2^m$ such that $\mathbf{z}_1 = \mathbf{z}0$ and $\mathbf{z}_1 = \mathbf{z}1$, it holds that $f(\mathbf{z}_1) = f(\mathbf{z}_2) = 1$. Since by construction the points $\mathbf{z}_1$ and $\mathbf{z}_2$ are consecutive in the lexicographic ordering, the codeword Eval$(f)\notin S_{(1,\infty)}$. 
\end{proof}
Now consider any sequence of RM codes, $\{{\mathcal{C}}_m(R)\}$, where ${\mathcal{C}}_m(R) = \text{RM}(m,r_m)$, where $r_m$ is as in \eqref{eq:rmval}. By Theorem \ref{thm:Reeves}, we know that the sequence $\{\mathcal{C}_m(R)\}$ achieves a rate of $R\in (0,C)$, where $C$ is the capacity of the unconstrained BMS channel. From Lemma \ref{lem:necesscond}, it follows that in order to obtain an upper bound on the number of codewords, Eval$(f) \in {\mathcal{C}}_m(R)$, that respect the $(1,\infty)$-RLL constraint, it is sufficient to obtain an upper bound on the number of polynomials $h$, for a given $g$, such that supp$(\text{Eval}(g)) \subseteq \text{supp}(\text{Eval}(h))$.

The following two lemmas from the literature will be useful in the proof of Theorem \ref{thm:rmub}.
\begin{lemma}[\cite{abbe1}, Lemma 36]
	\label{lem:short}
	Let $\mathcal{V} \subseteq \mathbb{F}_2^m$ be such that $|\mathcal{V}|\geq 2^{m-u}$, for some $u\in \mathbb{N}$. Then, it holds that
	\[
	\text{rank}(G(m,r)[\mathcal{V}])> {m-u \choose \leq r},
	\]
	where $G(m,r)$ is the generator matrix of the RM$(m,r)$ code, and $G(m,r)[\mathcal{V}]$ denotes the set of columns of $G(m,r)$, indexed by $\mathcal{V}$.
\end{lemma}
The lemma below follows from Proposition~1.6 in \cite{sam}:
\begin{lemma}
	\label{lem:sam}
	Let $\{\mathcal{C}_n\}$ be a sequence of Reed-Muller codes of blocklength $n$, rate $R_n$, and  having weight distribution $\left(A_{\mathcal{C}_n}(w): 0\leq w\leq n\right)$. If $\lim_{n\to \infty} R_n = R$, then 
	\[
	A_{\mathcal{C}_n}(w) \leq 2^{o(n)} \cdot \exp_2\left(2 \cdot  \ln \biggl(\frac{1}{1-R}\biggr) \cdot w \right)
	\]
	where $o(n)$ denotes a term $a_n$ such that $\lim_{n\to \infty}\frac{a_n}{n} = 0$.
\end{lemma}
We now provide the proof of Theorem \ref{thm:rmub}.
\begin{proof}[Proof of Theorem~\ref{thm:rmub}]
	Fix the sequence $\{{\mathcal{C}}_m(R) = \text{RM}(m,r_m)\}$ of RM codes, with
	\[
		r_m = \max \left\{\left \lfloor \frac{m}{2}+\frac{\sqrt{m}}{2}Q^{-1}(1-R)\right \rfloor,0\right\},
	\]
	where $R\in (0,R^*)$, for $R^*$ as in the statement of the theorem.
	As explained in the discussion, any Boolean polynomial $f$ whose evaluation Eval$(f)\in {\mathcal{C}}_m$ can be expressed as
	\[
	f(x_1,\ldots,x_m) = g(x_1,\ldots,x_{m-1}) + x_m\cdot h(x_1,\ldots,x_{m-1}), 
	\]
	where $g,h$ are such that Eval$(g)\in \text{RM}(m-1,r)$ and Eval$(h)\in \text{RM}(m-1,r-1)$. Now, for any codeword Eval$(g) \in \text{RM}(m-1,r_m)$ of weight $w$, we shall first calculate the number, $N_w(g)$, of codewords Eval$(h) \in \text{RM}(m-1,r_m)$ such that supp$(\text{Eval}(g)) \subseteq \text{supp}(\text{Eval}(h))$. Note that $N_w(g)$ serves as an upper bound on the number of codewords Eval$(h) \in \text{RM}(m-1,r_m-1)$ such that the same property holds.
	
	To this end, suppose that for any weight $w$, the integer  $u=u(w)$ is the smallest integer such that wt$(\text{Eval}(g)) = w \geq2^{m-1-u}$. Note that for any polynomial $g$ as above of weight $w$, the number of codewords in the code produced by shortening RM$(m-1,r_m)$ at the indices in supp$(\text{Eval}(g))$ indeed equals $N_w(g)$. Now, since dim$(\text{RM}(m-1,r_m)) = {m-1\choose \leq r_m}$, and from the fact that $w\geq 2^{m-1-u}$, we obtain by an application of Lemma \ref{lem:short} and the rank-nullity theorem, that
	\begin{align}
		\label{eq:nw}
	N_w(g)&\leq \text{exp}_2\left({m-1\choose \leq r_m} - {m-1-u \choose \leq r_m}\right)\\
	&=:M_{u(w)}. \notag
	\end{align}
	
	Let $\mathsf{R}_{{\mathcal{C}},m}^{(1,\infty)}(R)$ be the rate of the largest $(1,\infty)$-subcode of ${\mathcal{C}}_m(R)$, i.e.,
	\[
	\mathsf{R}_{{\mathcal{C}},m}^{(1,\infty)}(R) = \frac{\log_2(|\mathcal{H}_{(1,\infty)}^{(m)}|)}{2^m},
	\]
	where we let $\mathcal{H}_{(1,\infty)}^{(m)}$ denote the largest $(1,\infty)$-subcode of ${\mathcal{C}}_m(R)$. Then, it holds that
	\begin{align}
|\mathcal{H}_{(1,\infty)}^{(m)}|&\leq \sum_{g: \text{Eval}(g)\in \text{RM}(m-1,r_m)}{N_w(g)} \notag	\\
		&\stackrel{(a)}{\leq} \sum\limits_{w=2^{m-1-r_m}}^{2^{m-1}}A_{m-1,r_m}(w){M_{u(w)}} \notag\\
		&\stackrel{(b)}{\leq} \left\{\sum\limits_{w=2^{m-1-r_m}}^{2^{m-2}}A(w){M_{u(w)}}\right\}+\frac{1}{2}\cdot \text{exp}_2\left(m-1\choose \leq r_m\right)\times \notag\\
		&\ \ \ \ \ \ \ \ \ \ \ \ \ \ \ \ \ \ \ \ \ \ \ \  \ \text{exp}_2\left({m-1\choose \leq r_m} - {m-2\choose \leq r_{m}}\right)\notag \\
		&\stackrel{(c)}{\leq} \left\{\sum\limits_{w=2^{m-1-r_m}}^{2^{m-2}}A(w){M_{u(w)}}\right\}+
		\frac{1}{2}\times\notag\\
		&\ \ \ \ \ \ \ \ \ \ \ \ \ \ \ \ \ \ \ \ \ \ \ \  \ \text{exp}_2\left({m-1\choose \leq r_m}+{m-2\choose \leq r_{m}}\right)\notag\\
		&\stackrel{(d)}{\leq} \Bigg\{\sum\limits_{i=1}^{{r_m-1}}A([2^{m-2-i}: 2^{m-1-i}])\times \notag\\
		&\ \ \ \ \ \ \ \ \ \ \ \   \text{exp}_2\left({m-1\choose \leq r_m} - {m-2-i \choose \leq r_m}\right)\Bigg\}+ \notag\\
		 &\ \ \ \ \ \ \ \ \ \ \ \ \ \ \ \ \ \ \ \ \ \frac{1}{2}\cdot\text{exp}_2\left({m-1\choose \leq r_m}+{m-2\choose \leq r_{m}}\right) \label{eq:sub1},
	\end{align}
where $A([a:b])$ is short for $\sum_{i=a}^{b}A(w)$. Here,
\begin{enumerate}[label = (\alph*)]
	\item follows from equation \eqref{eq:nw}, and
	\item holds due to the following fact: since the all-ones codeword $\mathbf{1}$ is present in RM$(m-1,r_m)$, it implies that $A_{m-1,r_m}(w) = A(w) = A(2^{m-1}-w)$, i.e., that the weight distribution of codewords is symmetric about weight $w=2^{m-2}$. Therefore,
	\begin{equation}
		\label{eq:Aw}
		A(w>2^{m-2})\leq \frac{1}{2}\cdot \text{exp}_2{m-1\choose \leq r_m}.
	\end{equation}
Next,
	\item follows from Pascal's rule in combinatorics that for any $n,k \in \mathbb{N}$ with $n>k$:
	\[
	{n-1\choose k}+{n-1\choose k-1} = {n\choose k}.
	\]
	Picking $n=m-1$ and $k=r_m$, we obtain that
	\[
	{m-1\choose r_m}-{m-2\choose r_m} = {m-2\choose r_m-1},
	\]
	and hence that
	\[
	{m-1\choose \leq r_m}-{m-2\choose \leq r_m} < {m-2\choose \leq r_m}, \ \text{and}
	\]
	
	\item holds again from equation \eqref{eq:nw}.
\end{enumerate}
It is clear that a simplification of equation (d) depends crucially on good upper bounds on the weight distribution function. Now, since the code RM$(m-1,r_m)$ is obtained by shortening the code RM$(m,r_m)$ at positions $\mathbf{z} = (z_1,\ldots,z_m)$ where $z_m=1$, we obtain that $A_{m-1,r_m}(w)\leq A_{m,r_m}(w)$. 
Thus, we can use the result in Lemma \ref{lem:sam}, to get
\begin{align}
	&A_{m-1,r_m}([2^{m-2-i}: 2^{m-1-i}]) \notag\\
	&\leq A_{m,r_m}([2^{m-2-i}: 2^{m-1-i}]) \notag\\ 
	&\leq 2^{o(2^m)}\cdot \sum_{w=2^{m-2-i}}^{2^{m-1-i}}2^{2 \cdot \ln \left(\frac{1}{1-R}\right)\cdot w} \notag\\
	&\leq 2^{o(2^m)}\cdot \text{exp}_2\left(2 \cdot \ln \left(\frac{1}{1-R}\right)\cdot 2^{m-1-i}\right) =: B_i(m) \label{eq:sub2}
\end{align}
Therefore, putting \eqref{eq:sub2} back in \eqref{eq:sub1}, we get that
\begin{align}
	&2^m\mathsf{R}^{(1,\infty)}_{{\mathcal{C}},m}(R)\notag \\
	& = \log_2|\mathcal{H}_{(1,\infty)}^{(m)}| \notag\\
	&\leq {m-1 \choose \leq r_m}+
	 \log_2\Bigg\{\frac{1}{2}\cdot \text{exp}_2{m-2 \choose \leq r_m}+\notag \\
	 &\ \ \ \ \ \ \ \ \ \ \ \ \ \ \ \ \ \ \ \ \ \ \ \ \ \ \ \ \sum_{i=1}^{r_m-1}B_i(m)\cdot \text{exp}_2\left( -{m-2-i \choose \leq r_m}\right)\Bigg\} \notag\\
	  &= {m-1 \choose \leq r_m}+\log_2\left(\alpha(m)+\beta(m)\right),\label{eq:subs3}
\end{align}
where we define
\begin{align}
	\alpha(m)&:= \frac12 \cdot \text{exp}_2{m-2 \choose \leq r_m}, \ \text{and} \notag\\
	\beta(m)&:= \sum_{i=1}^{r_m-1}B_i(m)\cdot \text{exp}_2\left( -{m-2-i \choose \leq r_m}\right). \label{eq:beta}
\end{align}
In Appendix A, we show that for all $\delta>0$ sufficiently small and for $m$ sufficiently large, it holds that
\begin{align}
\beta(m) \ & \le \ 2^{o(2^m)} \cdot \exp_2\Bigg( 2^{m-3}\Bigg[4\ln\left(\frac{1}{1-R}\right) - R(1-\delta)\Bigg]\Bigg) \notag \\ 
& =: \ \theta(m).
\label{eq:inter}
\end{align}
Now, using Lemma~\ref{lem:rate}, we have
\[
\lim_{m\to \infty} \frac{1}{2^m} {m-2 \choose \leq r_{m}} = \frac{R}{4}.
\]
Hence, for small $\delta>0$, and for $m$ large enough, it holds that
\[
{m-2 \choose \leq r_{m}}\leq (1+\delta)\cdot 2^{{m}-2}\cdot R.
\]
Therefore, we get that
\begin{align}
\alpha(m) \ \leq \ \text{exp}_2\left((1+\delta)\cdot2^{{m}-2}\cdot R\right) \ =: \ \eta(m). \label{eq:inter2}
\end{align}
Now, substituting \eqref{eq:inter} and \eqref{eq:inter2} in \eqref{eq:subs3}, we get that 
\begin{align}
2^m\mathsf{R}^{(1,\infty)}_{{\mathcal{C}},m}(R)\leq {m-1 \choose \leq r_m}+\log_2\left(\eta(m)+\theta(m)\right).
\label{eq:inter3}
\end{align}
Putting everything together, we see that
\begin{align}
	&{R}^{(1,\infty)}_{{\mathcal{C}}}(R) \notag\\ &= \limsup_{m\to \infty}\mathsf{R}^{(1,\infty)}_{{\mathcal{C}},m}(R)\notag\\
	&\leq \lim_{m\to \infty} \frac{1}{2^m} \left[{m-1 \choose \leq r_m}+\log_2\left(\eta(m)+\theta(m)\right)\right] \notag\\
	&\stackrel{(a)}{\leq} \lim_{m\to \infty} \frac{1}{2^m} \left[{m-1 \choose \leq r_m}+\log_2\left(2\cdot \theta(m)\right)\right] \notag\\
	&\stackrel{(b)} =\frac{R}{2} + \lim_{m\to \infty} \frac{1}{2^m} \cdot \log_2 \theta(m) \notag \\
	&= \frac{R}{2}+\frac{4  \ln\left(\frac{1}{1-R}\right)  - 
		R(1-\delta)}{8} \notag\\
	&= \frac{3R}{8}+\frac{1}{2}\ln \left(\frac{1}{1-R}\right)+\frac{R\delta}{8} \label{eq:final}.
\end{align}
Note that inequality (a) follows from the fact for any $R \in (0,1)$, $\eta(m)\leq \theta(m)$ holds for all sufficiently small $\delta > 0$. Further, equation (b) is valid because $\lim_{m\to \infty} \frac{1}{2^m} {m-1 \choose \leq r_m} = \frac{R}{2}$, by Lemma~\ref{lem:rate}. Since equation \eqref{eq:final} holds for all $\delta>0$ sufficiently small, we can let $\delta\to 0$, thereby obtaining that
\begin{align}
	\label{eq:finalf}
	{R}^{(1,\infty)}_{{\mathcal{C}}}(R) \le \frac{3R}{8}+\frac{1}{2}\ln \left(\frac{1}{1-R}\right).
\end{align}
It can be numerically verified that for $R\in (0,R^*)$, where $R^*$ is as in the statement of the theorem, it holds that the right-hand side of equation \eqref{eq:finalf} is less than $R$, thereby providing a non-trivial upper bound. Finally, we note that since $|\mathcal{H}_{(1,\infty)}^{(m)}|\leq |\mathcal{C}_m(R)|$, the trivial upper bound of ${R}^{(1,\infty)}_{{\mathcal{C}}}(R)\leq R$ holds for all $R\in (0,1)$.
\end{proof}

\section{Conclusion}
\label{sec:conclusion}
This work proposed explicit, deterministic coding schemes, without feedback, for binary memoryless symmetric (BMS) channels with $(d,\infty)$- or $(0,k)$-runlength limited (RLL) constrained inputs. In particular, explicit achievable rates were calculated by identifying specific subcodes, of a capacity-achieving (over the unconstrained BMS channel) sequence of codes, all of whose codewords obey the input constraint. Furthermore, an upper bound was derived on the rate of the largest $(1,\infty)$-RLL subcodes of the specific RM codes that we had used in our lower bounds.




There is much scope for future work in this line of investigation. Firstly, following the close relationship between the size of $(1,\infty)$-RLL subcodes and the weight distribution of RM codes established in this work, a more sophisticated analysis of achievable rates can be performed with the availability of better lower bounds on the weight distribution of RM codes. Likewise, sharper upper bounds on the weight distribution of RM codes will also lead to better upper bounds on the rate of any $(1,\infty)$-RLL subcodes of a specific sequence of RM codes. It also remains to be seen how permutations of the co-ordinates of RM codes affect the size of their $(d,\infty)$-RLL subcodes. Another area of exploration could be the identification of other explicit (but not necessarily deterministic) coding schemes, over this class of channels, whose rates are computable. Such advancements will also help lend insight into the capacity of input-constrained BMS channels without feedback, which is a well-known open problem.
\section*{Acknowledgements}
This work was supported in part by a Qualcomm Innovation Fellowship India 2020. The work of V.~A.~Rameshwar was supported by a Prime Minister's Research Fellowship, from the Ministry of Education, Govt. of India. The work of N.~Kashyap was supported in part by MATRICS grant \ MTR/2017/000368 from the Science and Engineering Research Board (SERB), Govt. of India. 

\appendices
\section{Proof of Inequality \eqref{eq:inter}}
In this section, we show that the inequality $\beta(m)\leq \theta(m)$ holds, for large $m$ and sufficiently small $\delta>0$, with $\beta(m)$ and $\theta(m)$ defined in equations \eqref{eq:beta} and \eqref{eq:inter}, respectively. We fix an $R\in (0,R^*)$, where $R^*$ is defined in Theorem \ref{thm:rmub}.

We start with the expression 
\begin{equation}
\beta(m) = 2^{o(2^m)} \cdot \sum_{i=1}^{r_m-1} \exp_2 \left(\ln \left(\frac{1}{1-R}\right)\cdot 2^{m-i} - {m-2-i \choose \leq r_m}\right)
\label{eq:beta:2}
\end{equation}
We will split the sum $\sum\limits_{i=1}^{r_m-1}$ into two parts: $\sum\limits_{i=1}^{t_m}$ and $\sum\limits_{i={t_m}+1}^{r_m-1}$, where $t_m := \lfloor m^{1/3} \rfloor$.  For $i \in [t_m+1:r_m-1]$, we have
$$
\ln \left(\frac{1}{1-R}\right) \cdot 2^{m-i} - {m-2-i \choose \leq r_m} \ \le \  \ln \left(\frac{1}{1-R}\right) \cdot 2^{m-t_m} \ = \ o(2^m).
$$
Thus, the contribution of each term of the sum $\sum\limits_{i={t_m}+1}^{r_m-1}$ is $2^{o(2^m)}$, and since there are at most $r_m = O(m)$ terms in the sum, the total contribution from the sum is $2^{o(2^m)}$. 

Turning our attention to $i \in [t_m]$, we write
\begin{align}
\ln & \left(\frac{1}{1-R}\right) \cdot 2^{m-i} - {m-2-i \choose \leq r_m} \notag \\ & \ \ \ \ = \ 2^{m-2-i} \cdot \left[ 4 \ln\left(\frac{1}{1-R}\right) - \frac{1}{2^{m-2-i}} {m-2-i \choose \leq r_m}\right]. \label{i_in_tm_ineq}
\end{align}
By Lemma~\ref{lem:rate}, we obtain that $\frac{1}{2^{m-2-i}} {m-2-i \choose \leq r_m}$ converges to $R$ for all $i \in [t_m]$. In fact, with a bit more effort, we can show that this convergence is uniform in $i$. 
Note that, since $i \le t_m$, by virtue of \eqref{eq:rm_diff}, we have $|r_m - r_{m-2-i}|  \le  \frac{t_m+2}{2} + \frac{\sqrt{t_m+2}}{2} \lvert Q^{-1}(1-R)\rvert + 1 =: \nu_m$. Using the notation in the proof of Lemma~\ref{lem:rate}, we have $\frac{1}{2^{m-2-i}} {m-2-i \choose \leq r_m} = \Pr[S_{m-2-i} \le r_m]$. Thus, analogous to \eqref{sandwich}, we have for all sufficiently large $m$,
\begin{align*}
& \Pr[\overline{S}_{m-2-i} \le Q^{-1}(1-R) - \frac{\nu_m}{\frac12 \sqrt{m-2-t_m}}] \\
& \ \ \ \ \ \le \ \Pr[S_{m-2-i} \le r_m]  \\
& \ \ \ \ \ \ \ \ \ \ \ \le \ \Pr[\overline{S}_{m-2-i} \le Q^{-1}(1-R) + \frac{\nu_m}{\frac12 \sqrt{m-2-t_m}}].
\end{align*}
Now, we apply the Berry-Esseen theorem (see e.g., \cite[Theorem~3.4.17]{durrett}) which, in this case, asserts that $\left\lvert\Pr[\overline{S}_m \le x] - \Pr[Z \le x]\right\rvert \le 3/\sqrt{m}$, for all $x \in \mathbb{R}$ and positive integers $m$, where $Z \sim N(0,1)$. Thus, $\left\lvert\Pr[\overline{S}_{m-2-i} \le x] - \Pr[Z \le x]\right\rvert \le \frac{3}{\sqrt{m-2-i}} \le \frac{3}{\sqrt{m-2-t_m}}$ holds for all $x \in \mathbb{R}$ and $i \in [t_m]$. This yields
\begin{align*}
& \Pr[Z \le Q^{-1}(1-R) - \frac{\nu_m}{\frac12 \sqrt{m-2-t_m}}] - \frac{3}{\sqrt{m-2-t_m}} \\
& \ \ \le \ \Pr[S_{m-2-i} \le r_m]  \\
& \ \ \ \ \le \ \Pr[Z \le Q^{-1}(1-R) + \frac{\nu_m}{\frac12 \sqrt{m-2-t_m}}] + \frac{3}{\sqrt{m-2-t_m}}.
\end{align*}
Since $t_m$ and $\nu_m$ are both $o(\!\!\sqrt{m})$, we deduce that, as $m \to \infty$, $\Pr[S_{m-2-i} \le r_m] = \frac{1}{2^{m-2-i}} {m-2-i \choose \leq r_m}$ converges to $R$ uniformly in $i \in [t_m]$. 

Hence, for any $\delta \in (0,1)$ and $m$ large enough, it holds for all $i \in [t_m]$ that
\[
\frac{1}{2^{m-2-i}} {m-2-i \choose \leq r_m} \geq (1 - \delta) R,
\]
so that, carrying on from \eqref{i_in_tm_ineq}, 
$$
\ln \left(\frac{1}{1-R}\right) \cdot 2^{m-i} - {m-2-i \choose \leq r_m} \ \le \ 2^{m-3} \cdot \left[ 4 \ln\left(\frac{1}{1-R}\right) - (1-\delta)R \right]. \label{tm_ineq}
$$

To put it all together, recall that we split the sum $\sum\limits_{i=1}^{r_m-1}$ in \eqref{eq:beta:2} into two parts: $\sum\limits_{i=1}^{t_m}$ and $\sum\limits_{i={t_m}+1}^{r_m-1}$, where $t_m := \lfloor m^{1/3} \rfloor$. For sufficiently small $\delta > 0$, and all sufficiently large $m$, the contribution from the first sum is at most 
$$
m^{1/3} \cdot \exp_2\left(2^{m-3} \cdot \left[ 4 \ln\left(\frac{1}{1-R}\right) - (1-\delta)R \right] \right), 
$$
while that from the second sum is $\exp_2(o(2^m))$. Therefore, the overall sum $\sum\limits_{i=1}^{r_m-1}$ in \eqref{eq:beta:2} can be bounded above, for all sufficiently large $m$, by 
$$
2 m^{1/3} \cdot \exp_2\left(2^{m-3} \cdot \left[ 4 \ln\left(\frac{1}{1-R}\right) - (1-\delta)R \right] \right),
$$
Consequently,
$$
\beta(m) \ \le \ 2^{o(2^m)} \cdot \exp_2\Bigg( 2^{m-3}\Bigg[4\ln\left(\frac{1}{1-R}\right)  - R(1-\delta)\Bigg]\Bigg) \ =: \ \theta(m).
$$


\ifCLASSOPTIONcaptionsoff
  \newpage
\fi



%
\bibliographystyle{IEEEtran}
{\footnotesize
	\bibliography{references}}

%

\begin{IEEEbiographynophoto}{V.~Arvind Rameshwar}
Biography text here.
\end{IEEEbiographynophoto}

\begin{IEEEbiographynophoto}{Navin Kashyap}
Biography text here.
\end{IEEEbiographynophoto}





\end{document}